\documentclass[aps,superscriptaddress,11pt,a4paper,twoside,notitlepage,nofootinbib]{revtex4}

\usepackage{color,amsmath,enumerate,anysize,amsfonts,amsthm,graphicx,mathrsfs,amssymb}

\newcommand{\real}[0]{\mathbb{R}}

\newcommand{\NN}[0]{\mathbb{N}}

\newcommand{\bvec}[1]{{\bf#1}}
\theoremstyle{plain} \newtheorem{thm}{Theorem}
\theoremstyle{plain} \newtheorem*{thm'}{Theorem 7$'$}
\theoremstyle{plain} 
\theoremstyle{definition} \newtheorem*{rmk}{Remark}
\theoremstyle{definition} 
\theoremstyle{plain} \newtheorem{lemma}[thm]{Lemma}
\theoremstyle{plain} \newtheorem*{lem2}{Lemma 2}
\theoremstyle{plain} \newtheorem*{lem3}{Lemma 3}
\theoremstyle{definition} 
\theoremstyle{plain} \newtheorem{cor}[thm]{Corollary}

\marginsize{2.5cm}{2.5cm}{2cm}{1.5cm}

\begin{document}

\title{Measurement entropy in Generalized Non-Signalling Theory cannot detect bipartite non-locality}

\author{Josh Cadney}
\email{josh.cadney@bristol.ac.uk}
\affiliation{Department of Mathematics, University of Bristol, Bristol BS8 1TW, U.K.}

\author{Noah Linden}
\affiliation{Department of Mathematics, University of Bristol, Bristol BS8 1TW, U.K.}

\date{July 13, 2012}

\begin{abstract}
    We consider entropy in Generalized Non-Signalling Theory (also known as box world) where the most common definition of entropy is the measurement entropy. In this setting, we completely characterize the set of allowed entropies for a bipartite state. We find that the only inequalities amongst these entropies are subadditivity and non-negativity. What is surprising is that non-locality does not play a role - in fact any bipartite entropy vector can be achieved by separable states of the theory. This is in stark contrast to the case of the von Neumann entropy in quantum theory, where only entangled states satisfy $S(AB)<S(A)$.
\end{abstract}

\maketitle

\section{Introduction}

Entropy is a crucial concept in both classical and quantum information theory. The Shannon entropy was originally introduced as a measure of the uncertainty of a random variable \cite{Sha48}, which turned out to have many applications in information theory, including optimal compression rates and channel capacities. Remarkably, the von Neumann entropy was introduced 20 years before the Shannon entropy, and was motivated by thermodynamic considerations \cite{vN27}. It has found innumerable applications in quantum information theory, including its role as a measure of pure state entanglement \cite{BBPS96,PR97}, and as the analogue of the Shannon entropy in many quantum coding theorems \cite{Wilde}.

Given a multi-party quantum state $\rho$ one can compute the von Neumann entropy of its various reduced states e.g. $S(A):= S(\rho_A),\ S(AB):= S(\rho_{AB})$ etc., and so form the \emph{entropy vector} of this state $\rho$. So for example, for two-party states, the entropy vector is $(S(A),S(B),S(AB))$.  For $N$ parties, the entropy vector lives in the vector space of $2^N-1$ real dimensions.  The question of which vectors can arise has been the subject of increasing interest recently, both in the quantum (von Neumann entropy) \cite{Pip03,LW05,CLW12} and classical (Shannon entropy) \cite{YZ98,Mat07,DFZ11} cases.

For example for two parties, both quantum entropies $S_Q$ and classical entropies $S_C$ are non-negative and satisfy \emph{subadditivity}
\begin{eqnarray}
S_Q(AB)&\leq& S_Q(A) + S_Q(B),\nonumber\\
S_C(AB)&\leq& S_C(A) + S_C(B),
\end{eqnarray}
However the space of achievable entropy vectors is different for the classical and quantum cases.
Whereas classical entropies $S_C$ satisfy \emph{monotonicity}
\begin{equation}
S_C(AB)\ge S_C(A) ,\label{monotonicity}
\end{equation}
quantum entropies are more general and only satisfy the weaker Araki-Lieb inequality \cite{AL70}
\begin{equation}
S_Q(AB)\ge S_Q(A)-S_Q(B) .
\end{equation}
In particular, a vector such as $(1,1,0)$ is achievable as a quantum entropy vector; it is the entropy vector of a singlet.  However this vector is not achievable for any classical distribution; it does not satisfy (\ref{monotonicity}).  Thus the space of entropy vectors seems to capture some of the differences between classical and quantum states.  Indeed study of the space of achievable entropy vectors is a powerful tool in investigating multi-party entanglement of quantum states.

Mathematically, the space of entropy vectors is a cone \cite{Pip03}, and characterizing this cone is an important problem in classical and quantum information theory \cite{Yeung}.  The problem is completely solved for three or fewer parties in the classical and quantum cases \cite{YZ97,Pip03} (leading to different cones, of course); the cases of four or more parties (classical or quantum) remain open.  One may understand a cone either by giving the inequalities or, dually, by the extremal rays.  And, perhaps not surprisingly, points on these extremal rays typically correspond to interesting states.  For example for quantum entropy vectors of two parties the extremal rays include $\lambda(1,1,0),\ \lambda\ge 0$; a point on this ray may be achieved by the singlet, as mentioned above.  For three parties one of the extremal quantum rays may be achieved by the GHZ state (see also \cite{LW05,Ibi07}).

With these observations in mind we turn now to so called \lq\lq generalized probabilistic theories\rq\rq\ (GPTs) \cite{Bar07}. These are theories which generalize classical and quantum theories, beginning from an operational viewpoint, where states are characterized by the output distributions of certain measurements. One aim of this field of research is to compare these more general theories with quantum theory, and in doing so gain some intuition as to `why' Nature seems to prefer quantum theory.

Attempts have been made to introduce an entropy function within these general theories. The most popular seems to be the \emph{measurement entropy} which satisfies many desirable properties for an entropy function \cite{SW10}: it reduces to the Shannon and von Neumann entropies in classical and quantum theories respectively; it is always non-negative; and it is concave. In certain (quite broad) classes of theories, it is also subadditive and continuous.

Here we investigate features of the measurement entropy in `generalized non-signaling theory' (GNST) \cite{Bar07} (also known as \emph{box world} \cite{SB10}) - the most famous and well studied GPT, which allows all non-local correlations that are non-signalling. Our first goal is to characterize the set of \emph{entropy vectors}. It has already been noted that this entropy violates \emph{strong subadditivity} and so the allowed entropy vectors are in some sense more general than the corresponding classical and quantum ones \cite{SW10}. Our initial thought was that the space of achievable entropy vectors in GNST would reflect and shed light on the way this theory generalizes classical and quantum states.

We are able to completely determine the set of bipartite GNST entropy vectors (up to the closure). We find this set to be the cone in $\real^3$ cut out by the non-negativity and subadditivity of the entropy \emph{and no other inequalities}. This is in contrast with classical probability and quantum theory, where the analogous set is smaller due to the monotonicity (classical) and Araki-Lieb (quantum) inequalities.

What is very surprising, however, is that \emph{every} entropy vector in GNST can be achieved by a separable state. This means that the measurement entropy is unable to detect non-locality. This is not true in quantum theory, where all separable states (but certainly not all states) satisfy the monotonicity relation $S(A)\leq S(AB)$ \cite{HHH96}; thus one may say that those quantum entropy vectors that do not satisfy monotonicity are the \lq\lq truly\rq\rq\ quantum ones.

The structure of the paper is as follows: in section 2 we briefly review GPTs; in section 3 we consider in some detail the allowed measurements in GNST; in section 4 we characterize the set of bipartite GNST entropy vectors; and in section 5 we consider the implications of this result. We close the paper with some concluding remarks.

\section{Generalized Probabilistic Theories} \label{GPT}
It has been well known since Bell's theorem \cite{Bel64} that quantum theory admits correlations which are incompatible with any local classical theory. However, there are also correlations compatible with the no-signalling principle which cannot be produced by quantum theory. The most famous example is the PR-box \cite{PR94}.

\begin{center} \includegraphics[height=1.5in]{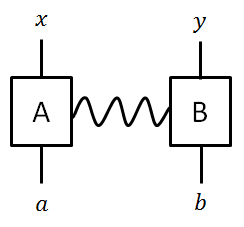} \end{center}
Here, two parties (Alice and Bob) each own part of a bipartite system. Alice and Bob choose inputs $x,y\in\{0,1\}$ respectively; the values of $x$ and $y$ correspond to different measurements on their systems. They then obtain outcomes $a,b\in\{0,1\}$ according to the distribution:
\begin{equation}
    p(a,b|x,y)=\left\{\begin{array}{cl}
        \frac12 & \text{ if } a\oplus b=xy \\
        0 & \text{ else}
    \end{array}\right.
\end{equation}
We know that quantum theory cannot produce such a distribution \cite{Tsi80}. The motivation behind GPTs is to consider what kinds of physical theory could admit these and other general no-signalling correlations.

In order to form a new physical theory, we assume that the state of a system is determined by the outcome distributions of measurements on the system.

For an individual system, we assume that there is a set of $k$ measurements, each with $l$ outcomes, which determine the state uniquely. We call these the \emph{fiducial} measurements. The state of the system is then a vector
\begin{equation}
    \bvec{p}=\left(\begin{array}{c}
        p(a=0|x=0) \\
        p(a=1|x=0) \\
        \vdots \\
        p(a=l-2|x=k-1) \\
        p(a=l-1|x=k-1) \\
    \end{array}\right)
\end{equation}
in a real vector space $V$. The values of $k$ and $l$ can vary from system to system (in the same way that different quantum systems have Hilbert spaces with different dimensions). For example, when $k=1$ there is only one fiducial measurement, and we say that the system is \emph{classical} since it is simply a classical random variable.

For a composite system, we make the further assumption that the fiducial measurements are those performed by simultaneously performing a fiducial measurement on each individual subsystem. This means that if the state is composed of $n$ individual systems, then the state can be considered as a vector, $\bvec{p}$, with components $p(a_1,\ldots, a_n|x_1,\ldots, x_n)$, also denoted $p(\bvec{a}|\bvec{x})$. Notice that $\bvec{p}$ naturally lives in the vector space $V_1 \otimes \ldots \otimes V_n$, where $V_i$ is the vector space containing the states of system $i$.

We now have many different types of system (each system contains some number, $n$, of individual systems, each of which has its own values for $k$ and $l$). We obtain a physical theory by specifying the sets of allowed states on each type of system. These sets must be convex to allow state mixing, and must satisfy the normalization condition:
\begin{equation}
    \sum_{\bvec{a}} p(\bvec{a}|\bvec{x}) = 1 \ \forall \bvec{x}
\end{equation}
Further, all states must satisfy the no-signalling constraints:
\begin{equation}
    \sum_{a_i} p(a_1,\ldots,a_i,\ldots,a_n|x_1,\ldots,x_i,\ldots,x_n) = \sum_{a_i} p(a_1,\ldots,a_i,\ldots,a_n|x_1,\ldots,x_i',\ldots,x_n)
\end{equation}
These constraints are important for two reasons: firstly, any state which violates these constraints would allow superluminal signalling. Secondly, they allow us to define the reduced state of a multipartite system:
\begin{equation}
    p(\hat{\bvec{a}}|\hat{\bvec{x}}) := \sum_{a_i} p(\bvec{a}|\bvec{x})
\end{equation}
where e.g. $\hat{\bvec{a}}=(a_1,\ldots,a_{i-1},a_{i+1},\ldots,a_n)$, and we know that the sum does not depend on the value of $x_i$.

These theories can then be extended to include general measurements (this will be discussed in more detail in section \ref{GNST}) and transformations, where we make further, physically motivated assumptions. For a full discussion see \cite{Bar07}.

\section{Generalized Non-Signalling Theory} \label{GNST}
In this paper we consider a particular GPT known as `generalized non-signalling theory' (GNST). GNST is the most general GPT in the sense that, for any type of system, the set of allowed states is \emph{all} those which satisfy no-signalling. It is also known as `box world' since we refer to individual systems in GNST as \emph{boxes}. In this section we are especially interested in the measurements which the theory permits.

\subsection*{Measurements in GNST}
Suppose we have a system with a set of allowed states $\mathscr{S}$. In a generalized probabilistic theory, an arbitrary measurement on the system (including, but not limited to, the fiducial measurements) has the following form: it is a set of pairs $(r,\mu_r)$, where $r$ is the outcome of the measurement, and $\mu_r$ is the corresponding \emph{effect}. An effect is a linear map $\mu :\mathscr{S}\to \left[ 0,1 \right] $ (so that $\mu_r(\bvec{p})$ is the probability that outcome $r$ is obtained when the measurement is performed on state $\bvec{p}$). To ensure that these probabilities always sum to 1, every measurement must have that $\sum_r \mu_r = u$, where $u$ is the constant map $u(\bvec{p})=1 \ \forall \bvec{p}\in\mathscr{S}$.

In GNST, any linear function $\mu : \mathscr{S}\to\left[ 0,1 \right]$ is an allowed effect, and any set of effects $\left\{ \mu_r \right\}$ which sum to the unit map is an allowed measurement. Here we review what is known about the set of measurements in GNST, and prove a slight generalization of a result in \cite{SB10} which we will use in Section \ref{result}.

Since effects are linear functionals, they must be of the form:
\begin{equation}
    \mu_r(p)=\sum_{\bvec{a},\bvec{x}} p(\bvec{a}|\bvec{x}) R_r(\bvec{a}|\bvec{x})
\end{equation}
for some vector $\bvec{R}_r$ (with entries indexed over $\bvec{a}$ and $\bvec{x}$). We say that $\bvec{R}_r$ represents $\mu_r$. However, note that there will be many vectors which represent each $\mu_r$.

The following lemma is crucial.
\begin{lemma}[{Barrett \cite[Appendix D]{Bar07}}] \label{effect}
    Every effect $\mu$ can be represented by some vector $\bvec{R}$ such that $0\leq R(\bvec{a}| \bvec{x}) \leq 1$ for all $\bvec{a},\bvec{x}$.
\end{lemma}

Now suppose we have a composite system of many boxes. One way in which we can perform a measurement is to do the following:
\begin{itemize}
    \item Choose one of the individual boxes and perform a fiducial measurement on that box.
    \item Based on the outcome of this measurement, choose another box and perform a fiducial measurement on that box.
    \item Repeat until all the boxes have been measured.
    \item Give the outcome of the measurement, $r$, which is a deterministic function of the outputs $\bvec{a}$.
\end{itemize}
We call a measurement which has this form a \emph{basic} measurement. The outcomes of a basic measurement are the values $r(\bvec{a})$. The probability of obtaining output $\hat{r}$ is equal to $\sum_\bvec{a} p(\bvec{a}|\bvec{x}(\bvec{a}))$ where the sum runs over all outputs $\bvec{a}$ such that $r(\bvec{a})=\hat{r}$. Here $\bvec{x}(\bvec{a})$ is the list of inputs which are entered in the measurement when outputs $\bvec{a}$ are obtained. Thus, a vector representing the effect $\mu_{\hat{r}}$ is $\bvec{R}_{\hat{r}}$ with components,
\begin{equation}
    R_{\hat{r}} (\bvec{a}|\bvec{x})=\left\{\begin{array}{cl}
        1 & \text{ if } r(\bvec{a})=\hat{r} \text{ and } \bvec{x}=\bvec{x}(\bvec{a}) \\
        0 & \text{ else}
    \end{array}\right.
\end{equation}

In \cite{SB10} it is shown that the only measurements which can be performed on systems of one or two boxes are basic measurements, or probabilistic mixtures of basic measurements. In section \ref{result} we will need the following generalization:

\begin{lemma} \label{twononclass}
    Let $A_1,A_2,B_1,\ldots,B_n$ be a system of boxes, where only $A_1$ and $A_2$ are not classical. Then all measurements on this system are basic, or mixtures of basic measurements.
\end{lemma}

The proof of this lemma is straightforward, and provided in the appendix.

\subsection*{Maximally informative measurements}
A notion which will be important in the next section is that of a \emph{maximally informative} or \emph{fine-grained} measurement. Let $M=\{ (r,\mu_r) \}$ and $N=\{ (s,\nu_s) \}$ be two measurements, and denote their sets of possible outcomes by $O_M$ and $O_N$ respectively. We say that $N$ is a \emph{refinement} of $M$ if $O_N$ can be partitioned into sets $P_r$ such that, for each $r$, $\mu_r = \sum_{s\in P_r} \nu_s$. In this case, $N$ can be used to perform $M$ (by performing $N$ and returning $r$ such that the outcome $s$ is in the set $P_r$). The refinement is \emph{trivial} if $\nu_s \propto \mu_r$ whenever $s \in P_r$. If $M$ has no non-trivial refinement, then no other measurement reveals strictly more information about the state, and hence we call $M$ maximally informative, or fine-grained.

The following lemma gives an important characterization of maximallly informative measurements in GNST. Although the result may seem obvious in view of lemma \ref{effect}, the proof requires a little effort and can be found in the appendix.
\begin{lemma} \label{fine-grained}
    A GNST measurement $M=\{(r,\mu_r)\}$ is maximally informative if and only if every effect $\mu_r$ can be represented by a vector with only one non-zero entry, which is between 0 and 1.
\end{lemma}

\begin{rmk}
    Lemma \ref{fine-grained} (together with lemma \ref{RequalsS}, found in the Appendix) ensures that basic measurements are maximally informative, if and only if the function $r$ is injective.
\end{rmk}

Suppose we have a composite system of subsystems $X$ and $Y$ (which may themselves be composite systems) and measurements $M_X$ and $M_Y$ on each system. One way to perform a measurement on the composite system would be to perform $M_X$ on $X$ and $M_Y$ on $Y$ independently. This is always an allowed measurement and we denote it $M_X \otimes M_Y$. Precisely, if $M_X$ has effects $\mu_r$ represented by vectors $\bvec{R}_r$ and $M_Y$ has effects $\nu_s$ represented by vectors $\bvec{R}_s$, then $M_X \otimes M_Y$ has effects $\tau_{r,s}$ which can be represented by vectors $\bvec{R}_r \otimes \bvec{R}_s$. (The latter $\otimes$ here is a genuine tensor product of vectors, as mentioned in section \ref{GPT}).

\begin{cor}\label{productmeas}
    If $M_X$ is a maximally informative measurement on a box $X$ and $M_Y$ is maximally informative on $Y$, then $M_X \otimes M_Y$ is a maximally informative measurement on the composite system $XY$.
\end{cor}
\begin{proof}
    If $\bvec{R}_X$ and $\bvec{R}_Y$ are vectors with one non-zero entry then so is $\bvec{R}_X \otimes \bvec{R}_Y$. Lemma \ref{fine-grained} then implies the result.
\end{proof}

\section{Entropy in GNST} \label{result}
We are now in a position to introduce the entropy function which we will study in this paper. The measure of entropy we use is the measurement entropy, $\hat{H}$, which is defined as follows. Suppose we have a state $\bvec{p}$ on a system $X$. Then,
\begin{equation} \label{mentropy}
    \hat{H}(X)_\bvec{p} := \inf_{M\in \mathscr{M}} H_M(X)_\bvec{p}
\end{equation}
where $H_M(X)_\bvec{p}$ is the Shannon entropy of the outcomes of measurement $M$ on system $X$ with state $\bvec{p}$, and the infimum is taken over $\mathscr{M}$, the set of all maximally informative measurements. When it is clear which state we are referring to we will omit the subscript $\bvec{p}$ from the notation. The motivation for such a definition of entropy comes from the fact that in classical and quantum theories it is none other than the Shannon and von Neumann entropies respectively. It has previously been studied in \cite{SW10,BBC10}. Lemma \ref{fine-grained} implies that in GNST the infimum in \eqref{mentropy} can be replaced by a minimum.

It is clear from the definition that the measurement entropy is always non-negative, since the Shannon entropy is non-negative. Before we can proceed with the main argument of the paper, we require two simple lemmas.

\begin{lemma}\label{subadd}
    Measurement entropy in GNST is subadditive - for any state of a joint system $XY$ we have that $\hat{H}(XY)\leq \hat{H}(X) + \hat{H}(Y)$.
\end{lemma}
\begin{proof}
    Suppose that $M_X$ is the measurement on system $X$ which achieves $\hat{H}(X)$, and $M_Y$ is the measurement on system $Y$ achieving $\hat{H}(Y)$. From section \ref{GNST} we know that $M:=M_X \otimes M_Y$ is a maximally informative measurement on system $XY$. Thus,
    \begin{align}
        \hat{H}(XY) &\leq H_M(XY) \\
            &\leq H_{M_X}(X) + H_{M_Y}(Y) \\
            &= \hat{H}(X) + \hat{H}(Y)
    \end{align}
    where the second inequality follows from the subadditivity of the Shannon entropy.
\end{proof}

\begin{rmk}
    Note that this proof applies to any GPT in which the analogue of Corollary \ref{productmeas} holds. This argument was presented in \cite{SW10}.
\end{rmk}

\begin{lemma} \label{product}
    If we restrict $\mathscr{M}$ in \eqref{mentropy} to include only basic measurements, then $\hat{H}$ is additive on product states.
\end{lemma}
\begin{proof}
    Let $p(\bvec{a},\bvec{b}|\bvec{x},\bvec{y})=p(\bvec{a}|\bvec{x})p(\bvec{b}|\bvec{y})$ be a product state of a joint system $XY$ (where $X$ is composed of $n$ boxes, and $Y$ is composed of $m$ boxes). We proceed by induction on $n+m$.

    Case $n+m=1$. Wlog $n=1,m=0$. Then $\hat{H}(XY)=\hat{H}(X)=\hat{H}(X)+\hat{H}(Y)$.

    Case $n+m\geq2$. Let $M$ be a measurement which achieves $\hat{H}(XY)$. Since $M$ is a basic measurement it must begin by performing a fiducial measurement on one of the individual boxes. Wlog assume that $M$ begins by performing measurement $x_1$ on box $X_1$. Denote the Shannon entropy of the outcome of this measurement by $H_{x_1}(X_1)$. Now suppose that we have performed the measurement, and the result $a_1$ is known. Let $\bvec{q}_{a_1,x_1}$ be the remaining distribution on $X_2\ldots X_n Y$. Then from the rules of conditional probability:
    \begin{equation}
        q_{a_1,x_1}(a_2,\ldots,a_n,\bvec{b}|x_2,\ldots,x_n,\bvec{y})= \frac{p(\bvec{a},\bvec{b}|\bvec{x},\bvec{y})}{p(a_1|x_1)}
    \end{equation}
    and notice that this is still a product state. Denote the remainder of the measurement, which has not yet been performed, by $M_{a_1}$. This is a basic measurement on $X_2\ldots X_n Y$. Then,
    \begin{align}
        \hat{H}(XY)_\bvec{p}&=H_M(XY)_{\bvec{p}}\\
            &=H_{x_1}(X_1)+ \sum_{a_1} p(a_1|x_1) H_{M_{a_1}}(X_2\ldots X_n Y)_{\bvec{q}_{a_1,x_1}} \label{groupline}
    \end{align}
    where the first line follows from the definition of $M$, and the second line from the grouping axiom of the Shannon entropy\footnote{Suppose that we partition the outcomes of a measurement into groups labelled $a_1,\ldots,a_k$ and break up the measurement into two stages. First observe variable $A$ - which group the outcome is in - and second observe variable $B$ - the outcome from among that group. The grouping axiom states that the entropy of the overall measurement is equal to $H(A)+\sum_i p(a_i)H(B|A=a_i)$.} \cite{Ash}, together with the fact that the outcome of $M$ is an injective function of the outputs $(\bvec{a},\bvec{b})$. From \eqref{groupline} we see that whenever $a_1$ can actually occur (i.e. $p(a_1|x_1)>0$), $M_{a_1}$ must be a measurement which achieves $\hat{H}(X_2\ldots X_nY)_{\bvec{q}_{a_1,x_1}}$, else it would be possible to achieve a lower value for $\hat{H}(XY)_\bvec{p}$. If $a_1$ cannot occur, we may just as well choose $M_{a_1}$ to be such a measurement. Therefore we have,
    \begin{align}
        \hat{H}(XY)_\bvec{p} &= H_{x_1}(X_1) + \sum_{a_1} p(a_1|x_1) \hat{H}(X_2\ldots X_n Y)_{\bvec{q}_{a_1,x_1}}\\
            &= H_{x_1}(X_1) + \sum_{a_1} p(a_1|x_1) \hat{H}(X_2\ldots X_n)_{\bvec{q}_{a_1,x_1}} + \hat{H}(Y)_\bvec{p} \label{grouped}
    \end{align}
    where the second line uses the induction hypothesis, and the fact that the reduced state of $\bvec{p}$ on system $Y$ is the same as that of $\bvec{q}_{a_1,x_1}$. By the grouping axiom, we see that the sum of the first two terms on the right hand side of \eqref{grouped} is the Shannon entropy of the outcomes of a measurement on system $X$. Therefore, by the definition of $\hat{H}$ it follows that
    \begin{equation}
        \hat{H}(XY)_\bvec{p}\geq \hat{H}(X)_\bvec{p} + \hat{H}(Y)_\bvec{p}
    \end{equation}
    Since the proof of lemma \ref{subadd} works also under the restriction to basic measurements, we arrive at the result.
\end{proof}

Our aim is to investigate the set of GNST entropy vectors. Let us focus on the two party case. We know that a two party entropy vector is a vector in $\real^3$, $(x,y,z)$, such that $x,y,z\geq0$ and $z\leq x+y$. Are there any further constraints on the values of $x,y,z$? We will show that in fact these conditions are all.

To this end, let $\mathscr{C}$ be the set of points given by our necessary conditions:
\begin{equation}
    \mathscr{C}:=\{ (x,y,z)\in\real^3 | x,y,z\geq0, z\leq x+y \}.
\end{equation}
Then $\mathscr{C}$ is a closed, convex cone (i.e. $v\in\mathscr{C}$ implies $\lambda v\in\mathscr{C}$ for all $\lambda\in\real_{\geq0}$, and whenever $v_1,v_2\in\mathscr{C}$, $v_1+v_2 \in\mathscr{C}$ also). This means that we can characterize $\mathscr{C}$ either by the linear inequalities which bound $\mathscr{C}$, or equivalently by its extremal rays. These extremal rays are the vectors:
\begin{align}
    e_1&:=(1,0,1)\\
    e_2&:=(0,1,1)\\
    e_3&:=(1,0,0)\\
    e_4&:=(0,1,0).
\end{align}
Any $v\in\mathscr{C}$ can be written as $v=\lambda_1 e_1 + \lambda_2 e_2 + \lambda_3 e_3 + \lambda_4 e_4$, with $\lambda_i \geq 0$ for all $i$.

Consider the following joint probability distribution:
\begin{equation}
    p(a,b)=\left\{\begin{array}{cl}
        \frac12 & \text{ if } b=0 \\
        0 & \text{ else}
    \end{array} \right.
\end{equation}
where $a,b$ both take values in $\{0,1\}$. Then we can consider $p(a,b)$ to be the state of a two box system, in which each box has only one possible input. The entropies are then just the Shannon entropies of the different reduced states. Hence the entropy vector is $(1,0,1)=e_1$. We can similarly find a probability distribution achieving $e_2$, and it is not hard to generalize these to distributions achieving $\lambda e_1$ and $\lambda e_2$ for any $\lambda\geq0$. (Indeed, consider the distribution of two random variables - one of which is deterministic, and the other with Shannon entropy $\lambda$).

Now consider a system of two boxes, $X,Y$, where $X$ has two possible inputs (0 and 1) and $N+1$ outputs ($0,1,\ldots,N$), and $Y$ is a random variable (i.e. only one input) with two possible outputs (0 and 1). The distribution $p(a,b|x)$ is as follows:
    \begin{equation} \label{mainexample0}
        p(a,b|x=0)=\left\{\begin{array}{cl}
            \frac{1}{2} & \text{if } a=b=0\\
            \frac{1}{2N} & \text{if } a\geq1,b=1\\
            0 & \text{else}
        \end{array}\right.
    \end{equation}
    \begin{equation} \label{mainexample1}
        p(a,b|x=1)=\left\{\begin{array}{cl}
            \frac{1}{2} & \text{if } a=0,b=1\\
            \frac{1}{2N} & \text{if } a\geq1,b=0\\
            0 & \text{else}
        \end{array}\right.
    \end{equation}

Since we have only two boxes, we know that the only allowed measurements are the basic measurements (or probabilistic mixtures of basic measurements, but a mixed measurement would not be optimal for achieving the minimum in \eqref{mentropy}). This makes it easy to calculate the entropies. The distribution of $X$ alone, for either input, is $(\frac12,\frac1{2N},\ldots,\frac1{2N})$ and hence $\hat{H}(X)=1+\frac12 \log N$. The reduced distribution of $Y$ is $(\frac12,\frac12)$ and so $\hat{H}(Y)=1$. Now consider the following measurement. First observe $Y$ to obtain output $b$. If $b=0$ set $x=0$, otherwise set $x=1$. Now observe $X$. With certainty we will find that $a=0$, and the distribution of the measurement outcomes is $(\frac12,\frac12)$. This implies that $\hat{H}(XY)\leq 1$, and in fact this measurement is optimal, i.e. $\hat{H}(XY)=1$.

We have discovered that for every $N$, $(1+\frac12\log N,1,1)$ is an entropy vector.

We now alter this scenario by adding an extra character, $\infty$, to the output alphabet of both boxes. Now for each $N\in\NN$ consider the following distribution:
    \begin{equation}
        p(a,b|x=0)=\left\{\begin{array}{cl}
            \frac{\lambda_N}{2} & \text{if } a=b=0\\
            \frac{\lambda_N}{2N} & \text{if } a\geq1,b=1\\
            1-\lambda_N & \text{if } a=b=\infty\\
            0 & \text{else}
        \end{array}\right.
    \end{equation}
    \begin{equation}
        p(a,b|x=1)=\left\{\begin{array}{cl}
            \frac{\lambda_N}{2} & \text{if } a=0,b=1\\
            \frac{\lambda_N}{2N} & \text{if } a\geq1,b=0\\
            1-\lambda_N & \text{if } a=b=\infty\\
            0 & \text{else}
        \end{array}\right.
    \end{equation}
where $\lambda_N$ are (as yet unspecified) constants (between 0 and 1). Fix a positive real value, $k$, and set $\lambda_N = \frac{2k}{\log N}$. By the same reasoning as above, this distribution has entropy vector $(\lambda_N + \frac{\lambda_N}{2}\log N + h(\lambda_N), \lambda_N + h(\lambda_N), \lambda_N + h(\lambda_N)) = (\lambda_N + k + h(\lambda_N), \lambda_N + h(\lambda_N), \lambda_N + h(\lambda_N))$, here $h(q):=-q\log q -(1-q)\log (1-q)$. As $N\to\infty$, $\lambda_N\to 0$ and $h(\lambda_N)\to 0$ and so we have found entropy vectors arbitrarily close to $(k,0,0)$.

\begin{thm} \label{maintheorem}
    Every vector in $\mathscr{C}$ is in the closure of the set of entropy vectors.
\end{thm}
\begin{proof}
    Consider an arbitrary vector in $\mathscr{C}$, $v=\lambda_1 e_1 + \lambda_2 e_2 + \lambda_3 e_3 + \lambda_4 e_4$. We have found states $\rho_1,\rho_2,\rho_3,\rho_4$ whose entropy vectors are (arbitrarily close to) $\lambda_1 e_1, \lambda_2 e_2, \lambda_3 e_3, \lambda_4 e_4$ respectively, and $\rho_1,\rho_2$ are entirely classical, whereas $\rho_3,\rho_4$ have only one non-classical box. If we take $\sigma=\rho_1\otimes\rho_2\otimes\rho_3\otimes\rho_4$, then $\sigma$ has only two non-classical boxes. By lemma \ref{twononclass} the only measurements on $\sigma$ are basic measurements, hence lemma \ref{product} tells us that the measurement entropy is additive on product states here. Consequently, the entropy vector of $\sigma$ is (arbitrarily close to) $v$.
\end{proof}

\section{Relation to Non-locality}
In the previous section we gave the proof of the main technical result of the paper, but we have not yet delivered the punch line. The alert reader would have noticed that all the states used in the proof of theorem \ref{maintheorem} are separable GNST states (and hence local in that they admit a local hidden variable description). For example, the state given by \eqref{mainexample0} and \eqref{mainexample1} can be decomposed in the following way:
\begin{equation}
    p(a,b|x)=\frac12 q_1(a|x) r_1(b) + \frac12 q_2(a|x) r_2(b)
\end{equation}
where $r_1(0)=1,r_1(1)=0,r_2(0)=0,r_2(1)=1$, and
\begin{equation}
    q_1(a|x)=\left\{\begin{array}{cl}
        1 & \text{if } a=x=0\\
        \frac{1}{N} & \text{if } a\geq1,x=1\\
        0 & \text{else}
    \end{array}\right.
\end{equation}
\begin{equation}
    q_2(a|x)=\left\{\begin{array}{cl}
        1 & \text{if } a=0,x=1\\
        \frac{1}{N} & \text{if } a\geq1,x=0\\
        0 & \text{else}
    \end{array}\right.
\end{equation}
Consequently, the theorem could just as easily have read:
\begin{thm'} \label{maintheorem'}
    Every bipartite GNST entropy vector is in the closure of the set of entropy vectors of separable GNST states.
\end{thm'}

Suppose that we are given a bipartite GNST state and told its entropy vector, which is known to be accurate to within some $\epsilon$, which can be arbitrarily small. Then the theorem tells us that we gain no knowledge of the non-local properties of the state. Whatever the entropy vector, the state may or may not be separable. This is in stark contrast with the von Neumann entropy, for which any vector with $S(AB)<S(A)$ instantly reveals that the state is entangled.

A further word of clarification. We showed that the measurement entropy in GNST satisfies only subadditivity and non-negativity and (at least for 2 parties) no further inequalities. (In fact we have evidence leading us to conjecture that this could also be true for 3 parties). This shows that entropy vectors in GNST are more general than entropy vectors in quantum theory. It would have been tempting to conclude that the reason for this is the extra non-locality available in GNST, but we have seen that this is not the case. What, then, is the cause?

Consider the following implementation of the state \eqref{mainexample0}-\eqref{mainexample1} using classical random variables.

\begin{center}\includegraphics[height=2in]{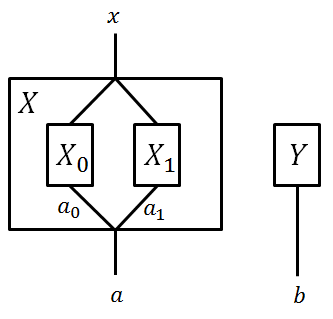}\end{center}

In this system there are three classical random variables $X_0,X_1$ and $Y$. $X_0$ and $X_1$ are concealed within box $X$ and are arranged such that:
\begin{itemize}
    \item If Alice inputs $x=0$, then $a=a_0$ and $X_1$ is destroyed.
    \item If Alice inputs $x=1$, then $a=a_1$ and $X_0$ is destroyed.
\end{itemize}
The distribution of $X_0,X_1,Y$ is as follows:
\begin{equation} \label{classdist}
    p(a_0,a_1,b)=\left\{\begin{array}{cl}
        \frac{1}{2N} & \text{if } a_0=0,a_1\geq1,b=0 \text{ or } a_0\geq1,a_1=0,b=1 \\
        0 & \text{else}
    \end{array}\right.
\end{equation}
which gives the same distribution as \eqref{mainexample0}-\eqref{mainexample1} for $p(a,b|x)$.

This raises an obvious question: if this GNST state can be realized via classical probability theory, why can't its entropies also be obtained there? The Shannon entropy vector of \eqref{classdist}, considered as a bipartite state, is $(1+ \log N, 1, 1+\log N)$ compared with the GNST entropy vector $(1+\frac12 \log N, 1, 1)$. But the Shannon entropy is none other than the measurement entropy in the classical setting. Since the same measurement that achieved $\hat{H}(XY)=1$ in GNST can also be performed classically, surely also $H(XY) \leq1$?

The reason this is not true is that, although this measurement can be performed in the classical setting, it is not \emph{maximally informative} there; since classically a maximally informative measurement must give the outputs of both $X_0$ and $X_1$. This mechanism by which GNST artificially makes measurements on classical random variables which are not maximally informative to be so, by hiding some of the variables within the boxes, is the reason that GNST entropy vectors are more general than classical ones.

\section{Discussion}
When faced with the task of naming his entropy, Shannon was apparently told by von Neumann to call it `entropy' because ``nobody knows what entropy really is, so in a debate you will always have the advantage''. Sixty years on this is still true. If `entropy' means `the minimum amount of uncertainty of a system under a maximally informative measurement' (and we accept the Shannon entropy as synonymous with `uncertainty') then we are forced to accept the measurement entropy as the unique entropy in any physical theory. However, this is not the prevailing definition. If we take a more pragmatic approach, and allow any function to be deemed an `entropy' if it satisfies a certain set of properties, then the question becomes: which properties do we choose?

This is where the result of the previous section fits. If you consider links between the von Neumann entropy and non-locality to be a happy coincidence, then this result has no bearing on the measurement entropy. If, however, you consider that in a highly non-local theory, such as GNST, an `entropy' ought to reflect this non-locality, then the measurement entropy cannot really be an `entropy'.

It would be interesting to explore the existence of a `better' entropy than measurement entropy in GNST. We know from \cite{AS11} that there is no function which obeys all the same desirable properties that the von Neumann entropy does in the quantum regime. However, could there be a function with the same desirable properties as the measurement entropy that also detects non-locality? Another interesting problem would be to determine whether or not the analogue of Theorem 7$'$ holds in other GPTs which are between quantum theory and GNST.

Ultimately, the Shannon and von Neumann entropies are useful functions, not because of the desirable properties they have, but because of their impact on physics and information theory. Their importance lies in the fact they can be used to give expressions for optimal rates of compression, or for channel capacities. To the best of our knowledge, only one such theorem is known using the measurement entropy \cite{SW10}. Other such theorems would be the best way to prove the usefulness of `entropy' measures.

\section*{Acknowledgments}
We thank Andreas Winter for very helpful discussions throughout the work, and especially Miguel Navascues who drew our attention to the separability of the states in section 4. The authors acknowledge support by the U.K.~EPSRC and the European Commission (STREP project ``QCS'').

\bibliographystyle{unsrt}
\bibliography{References}

\appendix
\section{} \label{appendix}

In this appendix we prove two of the lemmas stated in section \ref{GNST}.

\begin{lem2}
    Let $A_1,A_2,B_1,\ldots,B_n$ be a system of boxes, where only $A_1$ and $A_2$ are not classical. Then all measurements on this system are basic, or mixtures of basic measurements.
\end{lem2}
\begin{proof}
    Let $\bvec{a}=(a_1,a_2)$ denote the outputs of the 2 non-classical boxes, and $\bvec{x}=(x_1,x_2)$ denote their inputs. Let $\bvec{b}=(b_1,\ldots,b_n)$ denote the outputs of the classical boxes (since these boxes have no choice of input, we omit their inputs from the notation). Let $M = \left\{ (r,\mu_r) \right\}$ be an arbitrary measurement on the system, with $\left\{ \bvec{R}_r \right\}$ a set of representing vectors given by Lemma \ref{effect}.

    For each $\bvec{b}$ let $\bvec{R}_r^{(\bvec{b})}$ be the vector with components $R_r^{(\bvec{b})}(\bvec{a}|\bvec{x}) = R_r(\bvec{a},\bvec{b}|\bvec{x})$. Now, for fixed $\bvec{\hat{b}}$, we claim that $\{ \bvec{R}_r^{(\bvec{\hat{b}})} \}$ represent a measurement, $M^{(\bvec{\hat{b}})}$ on the non-classical part of the system. To see this, note that whenever $p(\bvec{a}|\bvec{x})$ is a state on $A_1,A_2$, we have:
    \begin{equation}
        \sum_{r,\bvec{a},\bvec{x}} p(\bvec{a}|\bvec{x}) R_r^{(\bvec{\hat{b}})}(\bvec{a}|\bvec{x}) = \sum_{r,\bvec{a},\bvec{b},\bvec{x}} p(\bvec{a}|\bvec{x}) \delta_{\bvec{b}\bvec{\hat{b}}} R_r(\bvec{a},\bvec{b}|\bvec{x}) = 1
    \end{equation}
    where the last equality follows from the fact that $p(\bvec{a}|\bvec{x}) \delta_{\bvec{b}\bvec{\hat{b}}}$ is an allowed state of the overall system.

    Since $M^{(\bvec{\hat{b}})}$ is a measurement on a two box system, it must be a mixture of basic measurements \cite{SB10}. The following is, therefore, also a mixture of basic measurements on $A_1,A_2,B_1,\ldots,B_n$.
    \begin{enumerate}[(i)]
        \item Obtain outputs $\bvec{b}$.
        \item Perform measurement $M^{(\bvec{b})}$, obtaining outcome $r$.
        \item Declare $r$ to be the outcome of the overall measurement.
    \end{enumerate}
    But, in fact, this measurement is $M$, since for this measurement:
    \begin{align}
        \text{Prob} (r) &= \sum_{\bvec{b}} p(\bvec{b}) \text{Prob}(r|\bvec{b})\\
            &= \sum_{\bvec{b}} p(\bvec{b}) \sum_{\bvec{a},\bvec{x}} R_r^{(\bvec{b})}(\bvec{a}|\bvec{x}) p(\bvec{a}|\bvec{b},\bvec{x})\\
            &= \sum_{\bvec{b}} p(\bvec{b}) \sum_{\bvec{a},\bvec{x}} R_r^{(\bvec{b})}(\bvec{a}|\bvec{x}) \frac{p(\bvec{a},\bvec{b}|\bvec{x})}{p(\bvec{b})}\\
            &= \sum_{\bvec{a},\bvec{b},\bvec{x}} R_r(\bvec{a},\bvec{b}|\bvec{x})p(\bvec{a},\bvec{b}|\bvec{x})
    \end{align}
    which is the same as the probability of getting outcome $r$ in measurement $M$. Here, $p(\bvec{a}|\bvec{b},\bvec{x})$ is the probability of getting output $\bvec{a}$ from systems $A_1,A_2$ when we input $\bvec{x}$ given knowledge of the outputs $\bvec{b}$ from the classical boxes. This is equal to $p(\bvec{b})^{-1}p(\bvec{a},\bvec{b}|\bvec{x})$ by the rules of conditional probability and the no-signalling condition.
\end{proof}

\begin{lem3}
    A measurement $M=\{(r,\mu_r)\}$ is maximally informative if and only if every effect $\mu_r$ can be represented by a vector with only one non-zero entry, which is between 0 and 1.
\end{lem3}

In order to make the proof of this lemma more clear, we first introduce two simple lemmas.

\begin{lemma} \label{useful}
    Let $(\bvec{a_1},\bvec{x_1})$ and $(\bvec{a_2},\bvec{x_2})$ be output-input pairs of a GNST system, with not both $\bvec{a_1}=\bvec{a_2}$ and $\bvec{x_1}=\bvec{x_2}$. Then there exists an allowed state, $\bvec{p}$, such that $p(\bvec{a_1}|\bvec{x_1})=0$ and $p(\bvec{a_2}|\bvec{x_2})>0$.
\end{lemma}
\begin{proof}
    First suppose that $\bvec{a_1}\neq\bvec{a_2}$. Then we can choose $\bvec{p}$ to be the distribution:
    \begin{equation}
        p(\bvec{a}|\bvec{x})=\delta_{\bvec{a}\bvec{a_2}}
    \end{equation}
    where $\delta_{\bvec{x}\bvec{y}}$ is 1 if $\bvec{x}=\bvec{y}$ and $0$ otherwise.

    Now suppose that $\bvec{a_1}=\bvec{a_2}$. This means that we must have $\bvec{x_1}\neq\bvec{x_2}$. Suppose that $\bvec{x_1}$ and $\bvec{x_2}$ disagree in the $i$th entry. Let $\bvec{\bar{a}_1}$ be a vector of outputs which disagrees with $\bvec{a_1}$ only in the $i$th entry. We can then choose $\bvec{p}$ to be the distribution:
    \begin{equation}
        p(\bvec{a}|\bvec{x})=\left\{\begin{array}{cl}
            \delta_{\bvec{a}\bvec{\bar{a}_1}} & \text{ if $\bvec{x},\bvec{x_1}$ agree in $i$th entry} \\
            \delta_{\bvec{a}\bvec{a_1}} & \text{ else}
        \end{array}\right.
    \end{equation}
\end{proof}

\begin{lemma} \label{RequalsS}
    Suppose that $\bvec{R}$ and $\bvec{S}$ are vectors representing the effect $\mu$, such that $\bvec{R}$ has exactly one non-zero entry, and $\bvec{S}$ has no negative entries. Then $\bvec{R}=\bvec{S}$.
\end{lemma}
\begin{proof}
    Let $\bvec{d}=\bvec{R}-\bvec{S}$. Since $\bvec{R},\bvec{S}$ both represent the same effect, it must be the case that $\bvec{d}\cdot\bvec{p}=0$ for all states $\bvec{p}$. We aim to show that $\bvec{d}=\bvec{0}$.

    Let $(\bvec{a_1},\bvec{x_1})$ be such that $R(\bvec{a_1}|\bvec{x_1})>0$. Let $(\bvec{a_2},\bvec{x_2})$ be a distinct, but otherwise arbitrary, output-input pair. Note that $d(\bvec{a_2}|\bvec{x_2})\leq0$. Now choose $\bvec{p}$ according to the previous lemma, and notice that for this choice of $\bvec{p}$, $\bvec{d}\cdot\bvec{p}$ will be negative, unless $d(\bvec{a_2}|\bvec{x_2})=0$. But $\bvec{a_2},\bvec{x_2}$ were arbitrary, so in fact $d(\bvec{a_1}|\bvec{x_1})$ is the only possibly non-zero component of $\bvec{d}$. Finally, let $\bvec{p}$ be the distribution $p(\bvec{a}|\bvec{x})=\delta_{\bvec{a}\bvec{a_1}}$ and then $\bvec{d}\cdot\bvec{p}=0$ implies that, in fact, $\bvec{d}=\bvec{0}$.
\end{proof}

\begin{proof}[Proof of lemma 3]
    Suppose that $\mu$ is an effect which can be represented by a vector, $\bvec{R}$, with only one non-zero entry. Suppose also that $\mu=\sum_i \nu_i$ for some effects $\nu_i$. For any vectors $\bvec{S}_i$ which represent $\nu_i$, the vector $\sum_i \bvec{S}_i$ represents $\mu$. By lemma \ref{effect} we can choose the $\bvec{S}_i$ so that they have no negative entries. Then, by lemma \ref{RequalsS}, this means that $\bvec{R}=\sum_i \bvec{S}_i$. This implies that $\bvec{S}_i \propto \bvec{R}$, and hence $\nu_i \propto \mu$ for all $i$. Thus, if all the effects of a measurement can be represented in this way, then the measurement must be maximally informative.

    Conversely, suppose $\mu$ cannot be represented by such a vector. Let $\bvec{R}$ be a vector with entries between 0 and 1 which represents $\mu$. $\bvec{R}$ must have more than one non-zero entry. Let $\bvec{R}_1$ be the vector which shares $\bvec{R}$'s first non-zero entry, and has zeroes elsewhere, and let $\bvec{R}_2=\bvec{R}-\bvec{R}_1$. Then $\bvec{R}_1$ and $\bvec{R}_2$ both represent valid effects $\nu_1,\nu_2$ with $\nu_1 + \nu_2 = \mu$. Now, if $\nu_1 \propto \mu$ then for some constant $\lambda$ we have $\lambda \nu_1 = \mu$ and hence $\lambda \bvec{R}_1$ represents $\mu$. But by lemma \ref{RequalsS} this implies that $\lambda \bvec{R}_1 = \bvec{R}$, which is clearly false. Consequently, there must exist a non-trivial refinement of any measurement containing $\mu$.
\end{proof}

\begin{rmk}
    In the proof of lemma \ref{useful} (and hence also in lemmas \ref{fine-grained} and \ref{RequalsS}) we assumed for simplicity that each box has more than one possible output. In the (rather trivial) case that some boxes have only one output, lemmas \ref{useful} and \ref{RequalsS} do not hold. However, it is still possible to obtain lemma \ref{fine-grained} by similar reasoning.

    The key observation is the following. Suppose that we have a system of boxes, some of which have fixed outputs. We denote these boxes by $X$, their inputs $\bvec{x}$ and their outputs $\bvec{0}$. The remainder of the system, $Y$, has boxes with inputs $\bvec{y}$ and outputs $\bvec{b}$. We now show that we can reduce the theory to one on system $Y$ only. The no signalling constraints ensure that for any state $\bvec{p}$ of the system, for all $\bvec{x},\bvec{x'}$, $p(\bvec{0},\bvec{b}|\bvec{x},\bvec{y})=p(\bvec{0},\bvec{b}|\bvec{x'},\bvec{y})$. This implies that a vector $\bvec{R}$ represents an effect $\mu$ if and only if the vector $\bvec{R'}$ also represents $\mu$, where
    \begin{equation}
        R'(\bvec{0},\bvec{b}|\bvec{x},\bvec{y})=\left\{\begin{array}{cl}
            \sum_\bvec{x'} R(\bvec{0},\bvec{b}|\bvec{x'},\bvec{y}) & \text{ if } \bvec{x}=\bvec{0}\\
            0 & \text{ else}
        \end{array}\right.
    \end{equation}
    Thus the effect is essentially an effect on system $Y$: $R''(\bvec{b}|\bvec{y}):=R'(\bvec{0},\bvec{b}|\bvec{0},\bvec{y})$. We can now run the proofs of lemmas \ref{fine-grained},\ref{useful} and \ref{RequalsS} for this effect.
\end{rmk}

\end{document}